\documentclass[preprint,pre,floats,showkeys,aps,amsmath,amssymb,superscriptaddress,12pt]{revtex4-1}
\usepackage{graphicx}
\usepackage{natbib}
\setcitestyle{square,comma,numbers,sort&compress}
\usepackage{amsthm}
\usepackage[left = 1in, top=1.2in,right=1in,bottom=1.2in,a4paper]{geometry}
\usepackage{graphicx}
\usepackage{amsmath}
\usepackage{amssymb}
\usepackage{hyperref}
\usepackage{xcolor}
\usepackage[normalem]{ulem}

\newcommand{\ket}[1]{|#1\rangle}
\newcommand{\bra}[1]{\langle #1|}

\begin{document}

\title{Modifications of the Page Curve from Correlations within Hawking Radiation}

\author{Mishkat Al Alvi}
\email{alvi0016@d.umn.edu}
\affiliation{Department of Physics and Astronomy, University of Minnesota, Duluth, Minnesota 55812, USA}
\author{Mahbub Majumdar}
\email{where.is.mahbub@gmail.com, majumdar@bracu.ac.bd}
\affiliation{BRAC University, 66 Mohakhali, Dhaka 1212, Bangladesh}
\author{Md.\! Abdul Matin (deceased)}
\affiliation{Bangladesh University of Engineering \& Technology, Dhaka 1000, Bangladesh}
\author{Moinul Hossain Rahat} 
\email{mrahat@ufl.edu}
\affiliation{Institute for Fundamental Theory, Department of Physics, University of Florida, Gainesville, Florida 32611, USA}
\author{Avik Roy}
\email{aroy@utexas.edu}
\affiliation{Center for Particles and Fields, Department of Physics, University of Texas at Austin, Austin, Texas 78712, USA}

\begin{abstract}

We investigate quantum correlations between successive steps of black hole evaporation and investigate whether they might resolve the black hole information paradox. `Small' corrections in various models were shown to be unable to restore unitarity. We study a toy qubit model of evaporation that allows small quantum correlations between successive steps and reaffirm previous results. Then, we relax the `smallness' condition and find a nontrivial upper and lower bound on the entanglement entropy change during the evaporation process. This gives a quantitative measure of the size of the correction needed to restore unitarity. 
We find that these entanglement entropy bounds lead to a significant deviation from the expected Page curve. 

\end{abstract}

\maketitle

\section{Introduction}


In 1972 Bekenstein audaciously associated a thermodynamic entropy to black holes that was proportional to the horizon area \cite{bek-1,bek-2}. This picture was elaborated in \cite{penr,christ,hawk-1}, and a complete description of classical black hole mechanics was presented in \cite{b-c-h, hawk-1}. 

The conflict between classical gravity that asserted that black holes don't emit radiation, and the thermodynamic picture whereby they have a temperature, was resolved by Hawking in 1975 \cite{hawk-2}. Using what are now standard techniques in curved space quantum field theory, Hawking showed that a black hole radiates as a black body with a temperature of $\frac{\kappa}{2\pi }$. 

Hawking suggested a heuristic picture whereby pair production occurred around the horizon. One particle fell into the black hole and another particle -- the Hawking radiation -- escaped to asymtopia. The picture was given a concrete realization in terms of tunneling by Parikh and Wilczek \cite{parikh}.  



However, this produced a new paradox. The ingoing and outgoing Hawking radiation from pair production would be entangled.   Thus, as the black hole evaporated, the entanglement entropy of the outgoing radiation would steadily increase. The emitted Hawking radiation at the end of the evaporation would then be entangled with ``nothing." Thus a black hole that began in a pure state would end in a mixed state violating unitarity. This is the {\em black hole information paradox.} 




Suppose it were possible to transfer the entanglement between the outgoing radiation and radiation/matter inside the black hole, to the outgoing radiation that came out late. Then all of the information in the black hole could be carried out by late time Hawking radiation and the final state would be a pure state  -- a pure state of early outgoing radiation entangled with late outgoing radiation.    

It was believed that one mechanism to realize this picture are `small' correlations between the Hawking quanta.  Many `small' effects might conceivably, collectively add up and allow for all of the information to come out \cite{presk}. Using strong subadditivity, Mathur showed that in a simple model of Hawking pairs being Bell pairs, that small correlations were unable to decrease the entanglement entropy enough to preserve unitarity \cite{peda}.  Later, Mathur showed that small correlations between consecutive, local, Hawking pair emissions -- such that a Hawking pair was correlated to the next emitted Hawking pair -- was likewise unable to reduce the entanglement entropy of the Hawking radiation to zero \cite{infall}.  In \cite{mathur-plumberg}, three models of non-local correlations among the Hawking radiation were considered and shown to not appreciably decrease the outgoing Hawking radiation entanglement entropy. Giddings in \cite{gidd-4} presented a nonlocal qubit model that was unitary. In \cite{avery}, Avery gave a framework to describe both unitary and non-unitary models of black hole evaporation and gave an example of how a non-unitary model could through a large deformation be made to be unitary. 
    
In order to ensure that the horizon has  {\it no drama}, it has been argued that corrections to Hawking pairs should be small.  In the presence of small corrections, the {\it niceness} conditions would still hold at the horizon which would enable the equivalence principle to hold at the horizon.  Large $\mathcal{O}(1)$ Hawking pair corrections were argued to be able to restore unitarity, but would remove the {\it no drama} nature of the horizon \cite{polchinski:bhip,marolf2017black,peda,infall}. 
    
In this paper, we continue the study of small correlations in restoring unitarity.  We provide a more general analysis than \cite{peda}.  We allow arbitrary, small correlations between Hawking pairs.  We show that such corrections do not restore unitarity. We next relax the `smallness' condition as prescribed in \cite{peda} and find a nontrivial upper and lower bound on the entanglement entropy change. This formulation allows us to quantify the kind of corrections required to restore unitarity. 

We find that our nontrivial bounds on the entanglement entropy give a deviation from the expected Page Curve \cite{Page-1,Page-2,Page-3}. This leads us to conclude that the corrections dictated by our qubit formalism are not compatible with physical unitary evolution. This result is in line with a more general result earlier proved in \cite{gidd-4} -- that corrections in the form of admixtures of Bell pair states alone,  are insufficient to restore unitarity. 



The organization of the paper is as follows. Section \ref{sec-2} introduces the leading order formulation of the black hole information paradox. In section \ref{sec-4}, we describe the framework often used to try to resolve the paradox using `small' corrections.  In section \ref{sec-5}, a simple toy model supporting the results of the previous section is briefly analyzed and shown to be unable to halt the monotonic growth of entanglement entropy. Section \ref{sec-6} generalizes the arguments in section \ref{sec-4} and finds upper and lower bounds on the change of entanglement entropy $\Delta S$ during the evaporation process when large correlations between Hawking quanta are allowed.  Section \ref{sec-7} shows how the upper and lower bounds on $\Delta S$ deduced in Section \ref{sec-6} produce a modified Page Curve. The modified Page Curve hints that Bell pair corrections are insufficient to make the evaporation process unitary.  Section \ref{sec-9} summarizes our findings.

\section{Leading Order Formulation of the Black Hole Information Paradox}\label{sec-2}

In this section we review the formulation of the information paradox as presented in \cite{peda} and that was followed by others. The purpose is to familiarize the reader with the framework we will use. 

We make certain assumptions that will allow us to use local effective field theory to handle the quantum gravity effects that produce Hawking radiation. This means that given a quantum state on a spacelike slice, we can evolve forward to future spacelike slices using Hamiltonian evolution. The assumptions are:


\begin{enumerate}

\item The black hole geometry can be foliated by spacelike slices continuous at the horizon  and the physics on these slices is given by a local quantum field theory \cite{peda, exactly}.

\item Suppose a collapsing shell with state $\ket{\Psi}_M$ produces the black hole. Time evolution of the hole, will push this matter far along spacelike slices inside the hole such that it is far from where Hawking radiation is being produced. Thus $\ket{\Psi}_M$ will at most, weakly affect the Hawking pairs.  Thus, after $N$ Hawking pairs are radiated the black hole + radiation state will be
\begin{equation}
  \ket{\Psi} \approx \ket{\Psi}_M \otimes\ket{\Psi}_1 \otimes \ket{\Psi}_2 \otimes \ldots \otimes \ket{\Psi}_N,
\label{N-pair}  
\end{equation}
where $| \Psi \rangle_i$ is the state of the $i$th Hawking pair. Since (\ref{N-pair}) has been written as a direct product of the black hole state $|\Psi \rangle_M$ and the Hawking quanta $|\Psi \rangle_i$, the Hawking radiation will in general not contain any information about the internal state of the black hole system $M$.


\item  The stretching of spacelike slices will cause creation/annihalation operators on different slices to be linearly related.  The states will be related by  \cite{hawk-3,gidd-1},
\begin{equation}\label{pair}
\ket{\Psi}_{pair} = Ce^{\beta c^{\dagger} b^{\dagger}}\ket{0}_c \ket{0}_b . 
\end{equation} 
where $\beta$ is a $c$-number and $c^\dagger$ and $b^\dagger$ are respectively, creation operators of the ingoing Hawking particle that is {\it captured} by the black hole, and the outgoing Hawking particle that {\it blasts out} of the black hole. Here $\ket{0}$ represents the vacuum state. 

To linear order, 
\begin{equation}\label{LO}
  \ket{\Psi}_{pair} = \frac{1}{\sqrt{2}}\left(\ket{0}_c\ket{0}_b + \ket{1}_c\ket{1}_b\right).
\end{equation}
We will use this Bell pair state as an approximation of Hawking radiation. The ingoing and outgoing $(b,c)$ particles are maximally entangled.  The entanglement entropy of the subsystem of the outgoing or subsystem of the ingoing radiation, for a single pair state (\ref{LO}) is  $S_{ent} = \log 2$, and for $N$ pairs is $S_{ent} = N\log 2$

\end{enumerate}
This monotonically increasing entanglement entropy lies at the heart of the black hole information paradox. Any solution would presumably include a mechanism to stop the growth of the entanglement entropy and reduce it to zero. 


Backreaction and quantum gravity effects will likely modify the Hawking pair state (\ref{LO}). As long as the effective field theory picture of the horizon physics holds, these corrections are expected to be \emph{small}.  

Small corrections can however build up and lead to large departures.  For example, consider perturbing the horizontally polarized photon pure state $\rho = \ket{\rightarrow}\bra{\rightarrow}$ and creating an admixture of with density matrix 
$\hat{\rho}' = (1 - \epsilon) \ket{\rightarrow}\bra{\rightarrow} + \epsilon \ket{\uparrow}\bra{\uparrow}$, where $\epsilon \ll 1$. The  entanglement entropy of of the perturbed state is small and is given by 
\begin{equation}
S(\hat{\rho}') = -\big(\left(1-\epsilon\right)\log\left(1-\epsilon\right) + \epsilon\log\epsilon \big),
\end{equation}

The \emph{fidelity}, given by $F(\rho_1,\rho_2) = \mathrm{Tr}{\sqrt{\rho_1^{1/2}\rho_2 \rho_1^{1/2}}}$, measures the closeness of two states.
There is high fidelity between $\hat{\rho}$ and $\hat{\rho}'$, since
$F\left(\hat{\rho}, \hat{\rho}'\right)= \sqrt{1 - \epsilon} \rightarrow 1$. However, the fidelity between $2N$ unperturbed photons described by $\hat{\rho}^{\otimes N}$, and $2N$ \emph{perturbed} photons described by $\hat{\rho}'^{\otimes N}$ will become small as $N$ becomes large. The two states look less and less like each other, since 
$F\left(\hat{\rho}^{\otimes N}, \hat{\rho}'^{\otimes N}\right) = \left(1-\epsilon\right)^{N/2} \rightarrow 0$.

Thus, small corrections to the the Hawking state that arise from effects ranging from energy conservation \cite{parikh,parikh2,tunnel} to thermal corrections \cite{zhang1}, can produce very different final states.  It has been hoped that such corrections can lead to information leakage from a black hole. 


\section{Small Corrections to the Hawking State - Mathur's Bound} \label{sec-4}

Hawking particles cannot carry information out of the black hole if the Hawking pairs are always in the same state -- such as in (\ref{LO}).  We therefore consider a space of Hawking states, with dimensionality larger than one.  Binary sequences are a popular way to encode information.  We therefore consider a two dimensional space of Hawking states spanned by $V_2 = \{ S^{(1)}, S^{(2)} \}$ where for the $(n+1)$th emission 
\begin{eqnarray}
S^{(1)} & = & \frac{1}{\sqrt{2}} \Big(|0\rangle_{c_{n+1}} |0\rangle_{b_{n+1}} + |1\rangle_{c_{n+1}} |1\rangle_{b_{n+1}} \Big) \nonumber \\
S^{(2)} & = & \frac{1}{\sqrt{2}} \Big(|0\rangle_{c_{n+1}} |0\rangle_{b_{n+1}} -  |1\rangle_{c_{n+1}} |1\rangle_{b_{n+1}} \Big).
\label{2dspace}
\end{eqnarray}
A larger space such as $V_4 = \{S^{(1)}, S^{(2)}, S^{(3)}, S^{(4)} \}$ where  $S^{(3)} = |0\rangle_{c_{n+1}} |1\rangle_{b_{n+1}}$ and $S^{(4)} = |1\rangle_{c_{n+1}} |0\rangle_{b_{n+1}}$ could have been chosen as was done in \cite{infall}. However, for our purposes, sequences from $(V_2)^{\otimes N}$ possess enough complexity to encode the information in a black hole. 

The complete system consists of the matter $M$ inside a black hole, incoming Hawking particles $c_i$, and outgoing Hawking particles  $b_i$. Given a basis $|\psi_i\rangle$ for the $\{M, c \}$ system of black hole matter and ingoing radiation, and a basis $|\chi_i\rangle$ for the $b_i$ quanta, the full state upon diagonalization is  
\begin{equation}
|\Psi_{M,c}, \psi_b(t_n)\rangle = \sum_{i}C_{i} | \psi_i \rangle  |\chi_i \rangle.
\end{equation}

At the next time-step of evolution, the $b_i$ quanta move a distance $\mathcal{O}(M^{-1})$ from the hole. The next Hawking particle is emitted a time of order $\mathcal{O}(M^{-1})$ after the previous emission. Thus, the $b_i$ and $\{ b_{i+1}, c_{i+1}\}$ systems will to a good approximation be causally disconnected.   We therefore assume that the $b_i$ particle is not affected by the $(i+1)$th, and later emissions.  

The $(i+1)$th emission will change the black hole states to 
\begin{equation}
| \psi_i \rangle  \longrightarrow S^{(1)} | \psi_i^{(1)} \rangle + S^{(2)} |\psi_i^{(2)} \rangle, \quad \textrm{where} \quad \|| \psi_i^{(1)} \rangle \|^2 + \| |\psi_i^{(2)}\rangle \|^2 = 1.
\label{correlations}
\end{equation}

In (\ref{correlations}) the Hawking pair state is entangled with the black hole, in contrast to (\ref{N-pair}). This enables correlations between the Hawking quanta, the black hole state, and previously emitted Hawking $c$ quanta to exist. It has been hoped that such correlations will allow information to be carried out by Hawking quanta.   (In the leading order analysis a new Hawking pair is in the state $S^{(1)}$ such that $|\psi_i^{(1)}\rangle  = |\psi_i \rangle $ and $|\psi_i^{(2)}\rangle  = 0$, and there are no correlations.)

The complete state after the $(n+1)$th emission is, 
\begin{equation}
|\Psi_{M,c}, \psi_b(t_{n+1})\rangle 
= 
\sum_i C_i\Big[S^{(1)} |\psi_i^{(1)} \rangle + S^{(2)} | \psi_i^{(2)} \rangle \Big] |\chi_i \rangle
= \Lambda^{(1)}S^{(1)} + \Lambda^{(2)}S^{(2)},
\label{tensorprod}
\end{equation}
\noindent  where,  
\begin{equation}
\Lambda^{(1)}=\sum_i C_i | \psi_i^{(1)} \rangle | \chi_i \rangle ;  \quad \Lambda^{(2)}=\sum_i C_i | \psi_i^{(2)} \rangle | \chi_i \rangle.
\label{lambda}
\end{equation}

In (\ref{tensorprod}) we have tensored with $| \chi_i \rangle$, the state of the previously emitted $b$ quanta. This is because we have assumed that new $b$ emissions don't affect previously emitted $b$ quanta. 

\textcolor{black}{The $\{b_1,\ldots, b_n\}$ quanta are not affected by later emissions. Thus, their entanglement entropy, $S(\{b_1,\ldots, b_n\})  
        = \sum_i |C_i|^2 \log |C_i|^2$,  stays the same at time-step $t_{n+1}$.}

The entanglement entropy of the $(b_{n+1},c_{n+1})$ pair with the rest of the system is given by the density matrix of the $(b_{n+1},c_{n+1})$ system,
\begin{equation}
\hat{\rho}_{b_{n+1},c_{n+1}} = \left[
    \begin{array}{ccc}
    \langle\Lambda^{(1)}|\Lambda^{(1)}\rangle & \langle\Lambda^{(1)}|\Lambda^{(2)}\rangle\\
    \langle\Lambda^{(2)}|\Lambda^{(1)}\rangle & \langle\Lambda^{(2)}|\Lambda^{(2)}\rangle
    \end{array}
    \right] \quad \textrm{where} \quad \|\Lambda^{(1)}\|^2 + \|\Lambda^{(2)}\|^2 = 1
\end{equation}

Mathur defined a modification to the Hawking state to be `small' if \cite{peda}
\begin{equation}
\|\Lambda^{(2)}\|^2 <\epsilon \quad \textrm{where}\quad \epsilon \ll 1.
\end{equation}
In this framework of small corrections, Hawking pairs are predominantly produced in the $S^{(1)}$ state and are rarely produced in the $S^{(2)}$ state. Mathur showed that the entanglement entropy at each time-step increases by at least $\log 2 - 2\epsilon$, which is positive for small $\epsilon$,  \cite{peda}. Models with correlations between consecutive emissions \cite{infall} and other toy models as in \cite{mathur-plumberg}, have echoed the inability of small corrections in decreasing the entanglement entropy increase at each time step.

In the next section we shall illustrate Mathur's bound in a toy model that incorporates a more general correlation compared to the model presented in \cite{infall}.

\section{A Simple Toy Model with Arbitrary Corrections} \label{sec-5}

We now consider a model that allows all previously radiated quanta to modify the state of a new Hawking pair. Let the first Hawking pair emitted be a Bell pair state
\begin{equation}
\ket{\Psi_1} = \frac{1}{\sqrt{2}} \left( \ket{0}_b \ket{0}_c + \ket{1}_b \ket{1}_c \right)
\end{equation}
We assume that the Hilbert space of the newly evolved pair is spanned by  $\{\ket{00}, \ket{11} \}$. The  $n$ pair state is given by
\begin{equation}
\ket{\Psi_n} = \sum_{i=0}^{2^n - 1} a_i \ket{i}_b \ket{i}_c \quad \textrm{where} \quad \sum_{i=0}^{2^n - 1} |a_i|^2 = 1
\end{equation}
Here $\ket{i}_c$  and $\ket{i}_b$ denote the states of the $n$ ingoing and $n$ outgoing quanta respectively, and $i$ is the  $n$ bit representation of the integer $i$.  We build the state of the $(n+1)$th pair from the previous $n$ quanta and the $(n+1)$th quanta as the entangled state,
\begin{equation}
\ket{\Psi_{n+1}} = \sum_{i=0}^{2^n - 1} a_i \ket{i}_b \ket{i}_c \otimes \frac{1}{\sqrt{2}} \Big( e^{s_{i,n,0}} \ket{0}_b \ket{0}_c + e^{s_{i,n,1}} \ket{1}_b \ket{1}_c \Big)
\end{equation}
The term $\frac{1}{\sqrt{2}}e^{s_{i,n,j}}$ is the amplitude to observe the new pair in the state $\ket{j}_b \ket{j}_c$, given that the earlier pairs were given by $\ket{i}_b \ket{i}_c$. If the modification to the Bell state for the $(n+1)$th pair is small, then $\left| s_{i,n,j} \right|$ is a small positive number.  Normalization requires, 
\begin{equation}\label{norm}
\sum_{j=0}^{1} e^{2s_{i,n,j}} = 2.
\end{equation}
The entanglement entropy of the $n+1$ radiated quanta is hence given by (where $\log$ is taken have base $e$),
\begin{align}
S(n+1) &= - \sum_{i} \sum_{j} \left( \frac{a_i e^{s_{i,n,j}}}{\sqrt{2}} \right)^2 \log \left( \frac{a_i e^{s_{i,n,j}}}{\sqrt{2}} \right)^2 \nonumber \\
& = - \sum_{i} \sum_{j} a_i^2 e^{2s_{i,n,j}} \left( \log a_i + s_{i,n,j} - \frac{1}{2}\log 2 \right) \nonumber \\
& = -\sum_i a_i^2 \log a_i^2 + \log 2 - \sum_i a_i^2 \sum_j s_{i,n,j} e^{2s_{i,n,j}}.
\label{n+1}
\end{align}

The first term in (\ref{n+1}) is the entanglement entropy of the first $n$ quanta.  Upon defining $\Delta S = S(n+1) - S(n)$, we find that
\begin{equation}\label{ds}
\Delta S = \log 2 - \sum_i a_i^2 \sum_j s_{i,n,j} e^{2s_{i,n,j}}.
\end{equation}
We can find an upper and lower bound on $\Delta S$. Consider the quantity 
\begin{equation}
f_{i,n} = \sum_j s_{i,n,j} e^{2s_{i,n,j}} = s_{i,n,0}e^{2s_{i,n,0}} + s_{i,n,1} e^{2s_{i,n,1}},
\end{equation}
where $s_{i,n,0}, s_{i,n,1} \in \mathbb{R} $. Using Lagrange Multipliers we can easily see that $f_{i,n} \ge 0$ with equality at $s_{i,n,0} = s_{i,n,1}=0$. The normalization (\ref{norm}) requires that $s_{i,n,0} \le \frac{1}{2} \log 2 $ and $s_{i,n,1} \le \frac{1}{2} \log 2 $. Therefore, using (\ref{norm}) again, we find that $f_{i,n}  \le 2 \cdot \frac{1}{2} \log 2  = \log 2$.

We therefore find that,
\begin{equation}
0 \le \Delta S \le \log 2.
\end{equation}

Our calculation therefore gives a lower bound to $\Delta S$. This improves Mathur's lower bound of $\Delta S \ge \log 2 - 2 \epsilon$ which is positive only if $\epsilon$ is small.  Our calculation shows that even large corrections in this model cannot lead to entropy decrease.  This leads us to the more general model in the next section.


\section{Generalization of Mathur's Bound} \label{sec-6}
In this section we generalize Mathur's bound. Some of this work was presented in a preliminary form in \cite{icece}. We allow corrections of arbitrary magnitude to the leading order analysis. Although, arbitrary corrections may destroy the \emph{Solar System Limit} as described by Mathur, and expectations of no-drama at the horizon, our results establish interesting and nontrivial strong upper and lower bounds on $\Delta S$, which must be obeyed even if $\epsilon \sim 1$.  

To facilitate the derivation, we first derive two lemmas. Let us introduce the correction parameters,
\begin{eqnarray}
\langle\Lambda^{(2)}|\Lambda^{(2)}\rangle = \epsilon^2, 
&\quad & 
\langle\Lambda^{(1)}|\Lambda^{(1)}\rangle = 1-\epsilon^2 \nonumber \\ 
\langle\Lambda^{(1)}|\Lambda^{(2)}\rangle = \langle\Lambda^{(2)}|
\Lambda^{(1)}\rangle = \epsilon_2,
&\quad &
\gamma^2 = 1-4\Big(\epsilon^2(1-\epsilon^2)-\epsilon_2^2\Big).
\label{eqn:con3}
\end{eqnarray}
We define a correction to be small if $|\epsilon| \ll 1$.  Although $\epsilon_2$ is generally complex and $\langle\Lambda^{(1)}|\Lambda^{(2)}\rangle = \langle\Lambda^{(2)}|\Lambda^{(1)}\rangle^{*}$, for simplicity in (\ref{eqn:con3}) we assume $\epsilon_2$ to be real. Complexifying it gives the same result.

\newtheorem{lone}{Lemma}
\begin{lone}
The entanglement entropy of the newly created pair is given by \[S(p)\leq\sqrt{1-\gamma^2}\log2.\]
\end{lone}

\begin{proof}
The reduced density matrix for the pair is
\begin{equation}
\hat{\rho}_p = \left[
    \begin{array}{ccc}
    1-\epsilon^2 & \epsilon_2\\
    \epsilon_2 & \epsilon^2
    \end{array}
    \right].
\end{equation}
The eigenvalues of this matrix are: $\lambda_1=\frac{1+\gamma}{2}$ and $\lambda_2=\frac{1-\gamma}{2}$, where $\gamma = \sqrt{1-4(\epsilon^2(1-\epsilon^2) - \epsilon_2^2)}$. Thus $ 0 \le \gamma \le 1$. This implies that 
\begin{equation}
    0 \le \epsilon^2(1-\epsilon^2) - \epsilon_2^2 \le \frac{1}{4}.
\label{gamma:bound}
\end{equation}
Hence, the entanglement entropy of the pair is
\begin{eqnarray}
S(p) & = & -\mathrm{tr} \hat{\rho}_p \log \hat{\rho}_p = -\sum_{i=1}^2 \lambda_i \log \lambda_i \nonumber \\
    & = & \log 2 - \frac{1}{2}[(1+\gamma)\log(1+\gamma)+ (1-\gamma)\log(1-\gamma)].
\end{eqnarray}

It can be shown that for $0 \leq x \leq 1$,
\begin{equation}\label{label14}
(1-x^2)\log2 \leq \log2 - \frac{1}{2}[(1+x)\log(1+x)
+(1-x)\log(1-x)] \leq \sqrt{1-x^2}\log2.
\end{equation}
The result follows from (\ref{label14}).
\end{proof}

\newtheorem{lthree}[lone]{Lemma}
\begin{lthree}
\begin{equation}
(1-4\epsilon_2^2)\log 2 \leq S(b_{n+1})=S(c_{n+1})\!\leq\! \sqrt{1-4\epsilon_2^2}\log2 \nonumber
\end{equation}
\end{lthree}

\begin{proof}
The complete state of the system after the creation of $n+1$ pairs is
\begin{eqnarray}
|\Psi_{M,c}, \psi_b(t_{n+1})\rangle  &=& \Big[|0\rangle_{c_{n+1}}|0\rangle_{b_{n+1}}\frac{1}{\sqrt{2}}(\Lambda^{(1)}+\Lambda^{(2)})\Big]\nonumber\\
&&+\Big[|1\rangle_{c_{n+1}}|1\rangle_{b_{n+1}}\frac{1}{\sqrt{2}}\!(\!\Lambda^{(1)}-\Lambda^{(2)})\Big],
\end{eqnarray}
\noindent where  $\Lambda^{(1)}$ and  $\Lambda^{(2)}$ reflect the state of the black hole and are defined by  (\ref{lambda}).

Now, the reduced density matrix of the $c_{n+1}$ or $b_{n+1}$ quanta is
\begin{eqnarray}
\hat{\rho}_{b_{n+1}} &=& \hat{\rho}_{c_{n+1}}\nonumber \\
&=&\left[\begin{array}{cc}
\frac{1}{2}\langle(\Lambda^{(1)}\!+\Lambda^{(2)})|(\Lambda^{(1)}+\Lambda^{(2)})\rangle & 0 \\
0 & \frac{1}{2}\langle(\Lambda^{(1)}-\Lambda^{(2)})|(\Lambda^{(1)}-\Lambda^{(2)})\rangle
\end{array}\right]\nonumber\\
&=& \left[\begin{array}{cc}
\frac{1+2\epsilon_2}{2} & 0\\
0 & \frac{1-2\epsilon_2}{2}
\end{array}\right].
\end{eqnarray}
Then, entanglement entropy of the $c_{n+1}$ or $b_{n+1}$ quanta is
\begin{eqnarray}
S(b_{n+1})=S(c_{n+1}) &=& \log2 - \frac{1+2\epsilon_2}{2}\log(1+2\epsilon_2)\nonumber\\
        &&-\: \frac{1-2\epsilon_2}{2}\log(1-2\epsilon_2).
\end{eqnarray}
The result follows directly from (\ref{label14}).
\end{proof}

Using these lemmas we now prove the following theorem.
\vspace{2 mm}
\newtheorem{nthm}{Theorem}
\begin{nthm}
The change in the entanglement entropy from time-step $t_n$ to $t_{n+1}$ is restricted by the following bound:
\begin{equation}\label{label10}
1-4\epsilon_2^2-\sqrt{1-\gamma^2} \leq \frac{\Delta S}{\log 2} \leq \sqrt{1-4\epsilon_2^2}
\end{equation}
where $\Delta S = S(b_{n+1}, \{b\})-S(\{b\})$.
\label{theorem-bounds}
\end{nthm}

\vspace{1 mm}
\begin{proof}
Let us assume $A=\{b\}$ and $B=b_{n+1}$ and $C=c_{n+1}$. Using the \emph{strong subadditivity inequality}, $S(A)+S(C) \leq S(A,B)+S(B,C)$
\noindent
we find that,
\begin{eqnarray} \label{label8}
S(\{b\})+S(c_{n+1}) & \leq & S(\{b\},b_{n+1})+S(b_{n+1},c_{n+1}) \nonumber \\
\Rightarrow \Delta S 
& \geq &  (1-4\epsilon_2^2)\log 2 - \sqrt{1-\gamma^2}\log 2.
\end{eqnarray}
\noindent
Inequality (\ref{label8}) follows from using Lemma 1 and Lemma 2.

Now using the \emph{subadditivity inequality}, $S(A)+S(B) \geq S(A,B)$, we find,

\begin{eqnarray} \label{label9}
S(\{b\})+S(b_{n+1})  & \geq &  S(\{b\},b_{n+1}) \nonumber\\
\Rightarrow \Delta S & \leq & \sqrt{1-4\epsilon_2^2}\log 2
\end{eqnarray}
\noindent Inequality (\ref{label9}) follows from Lemma 2 and combining (\ref{label8}) and (\ref{label9}) gives Theorem \ref{theorem-bounds}.
\end{proof}
\vspace{1 mm}

The change in the entanglement entropy upper and lower bounds in Theorem \ref{theorem-bounds} hold irrespective of the size of the correction $\epsilon$. They hold in particular for $\mathcal{O}(1)$ corrections. In the given effective field theory framework, they thus give an upper bound for the change in entanglement entropy for non-trivial corrections and can help us understand the size of the correction needed to restore unitarity.



\subsection{Nontriviality}
When $\epsilon \rightarrow 0$, the Hawking radiation is always in the same state,  $S^{(1)}$. Thus the Hawking radiation can't carry out information from the black hole.  This is reflected by the upper and lower bounds in (\ref{label10}) being the same and thus the entropy increase is maximal. The case of $\epsilon \rightarrow 1$ is similar. This time all of the Hawking radiation is in the state $S^{(2)}$ and the entropy increase given by (\ref{label10}) is thus maximal.  Thus (\ref{label10}) gives a non-trivial bound for $\epsilon \neq 0$ and $\epsilon \neq 1$, particularly for large corrections.  The analysis in \cite{peda} and the various generalizations reviewed in \cite{avery} do not address the case of large corrections.  

Consider the maximum difference between the lower and upper bound in (\ref{label10}).  Denoting this quantity by $D_{\Delta S}$, we find
\begin{equation}\label{dds}
D_{\Delta S} = \log 2 \left(\sqrt{1-4\epsilon_2^2} + \sqrt{4\left(\epsilon^2(1-\epsilon^2) - \epsilon_2^2\right)} - (1-4\epsilon_2^2)\right)
\end{equation}
It is straightforward to show that for any value of $\epsilon_2$, $D_{\Delta S}$ is maximized when $\epsilon^2(1-\epsilon^2)$ is maximal, i.e. $\frac{1}{4}$. Then (\ref{dds}) reduces to
\begin{equation}
\frac{D_{\Delta S}}{\log 2} = 2\sqrt{1-4\epsilon_2^2} - (1-4\epsilon_2^2)
\end{equation}
It is straightforward to show that maximization of $D_{\Delta S}$ requires $\frac{dD_{\Delta S}}{d\epsilon_2} = 0 \Rightarrow \epsilon_2 = 0$. This can also be seen from (\ref{gamma:bound}). This gives,
\begin{equation}
\frac{D_{\Delta S}}{\log 2} \le 1.
\end{equation}
Thus $D_{\Delta S}$ never exceeds $\log 2$. While it is clear that that $\Delta S \le \log 2$, the result that $D_{\Delta S} \le \log 2$ is a different nontrivial bound.



\subsection{Small corrections are not enough}
In the beginning of the  black hole evaporation process the entanglement entropy of the outgoing pairs will be significant because the outgoing radiation will carry little information about the hole. If the evaporation is unitary then as more Hawking particles are emitted, more and more information about the hole will come out and the entanglement entropy will decrease \cite{Page-1}. Once all of the information has come out, the entanglement entropy will be zero as the $\{b \}$ system will be a pure state.  

This means that at some point the lower bound in (\ref{label10}) must become negative. The occurs when, 
    \begin{equation}\label{label11}
    4\epsilon_2^2 (1-4\epsilon_2^2) > 1-4\epsilon^2 (1-\epsilon^2).
    \end{equation}
The maximum of the left hand side of (\ref{label11}) is $\frac{1}{4}$. This  implies that
\begin{equation}
1-4\epsilon^2 (1-\epsilon^2) <  \frac{1}{4} \quad \Longrightarrow \quad \frac{1}{2}<\epsilon < \frac{\sqrt{3}}{2}.
\label{condn}
\end{equation}
This gives a necessary (but not sufficient) condition for unitarity. Equation (\ref{condn}) implies that relatively large corrections are needed. The bound in (\ref{label10}) does not require the correction parameters to be small and thus it remains of interest even in the large correction regime given by (\ref{condn}). This is a unique feature of the derivation presented in this paper and will contribute to the results in the following section.

\section{Incompatibility of Bell Pair State Correction to Page Curves}\label{sec-7}

In \cite{Page-1}, Page showed that the a small subsystem like the outgoing radiation emitted by a black hole at the beginning of evaporation \cite{Page-2}, is maximally entangled with the larger system.  Mathematical proofs were provided in \cite{Pageproof-1, Pageproof-2,Pageproof-3}. 

Suppose the size of the smaller subsystem's Hilbert space is $m$ and the larger subsystem's Hilbert space dimension is $n$, such that the combined system has dimension  $mn$. Suppose the combined system is described by a random pure state. Page showed that in this case the entropy $S_{m,n}$, of the smaller subsystem is
\begin{equation}
S_{m,n} = \log m - \frac{m}{2n} + \mathcal{O}\left(\left(m/n\right)^2\right)
\end{equation}
Thus for $n\gg m$, $S_{m,n}$ is the very near the course-grained entropy of the outgoing radiation, which is $\log{m}$. Also, the outgoing radiation is thermal .  

As the outgoing radiation entropy rises, the size of the outgoing radiation subsystem $m$, increases.  This means that information will leak out of the black hole with increasing time.   At some point enough information will have leaked out such that $m \approx n$ and $S_{m.n}$ will decrease.  Eventually, we will have $m \gg n$. The analysis will then be similar to the $n \gg m$ case.  Broadly speaking, the $m \gg n$ and $n \gg m$ cases describe similar systems, but with $m$ and $n$ interchanged.  The fully symmetric case for the Page Curve is shown in Figure \ref{fig1}a. Note, the turnover point, or Page Time, occurs when half of the black hole has evaporated. 


\begin{figure}[h]
\centering
\begin{tabular}{ccc} 
      \includegraphics[width=0.33\textwidth]{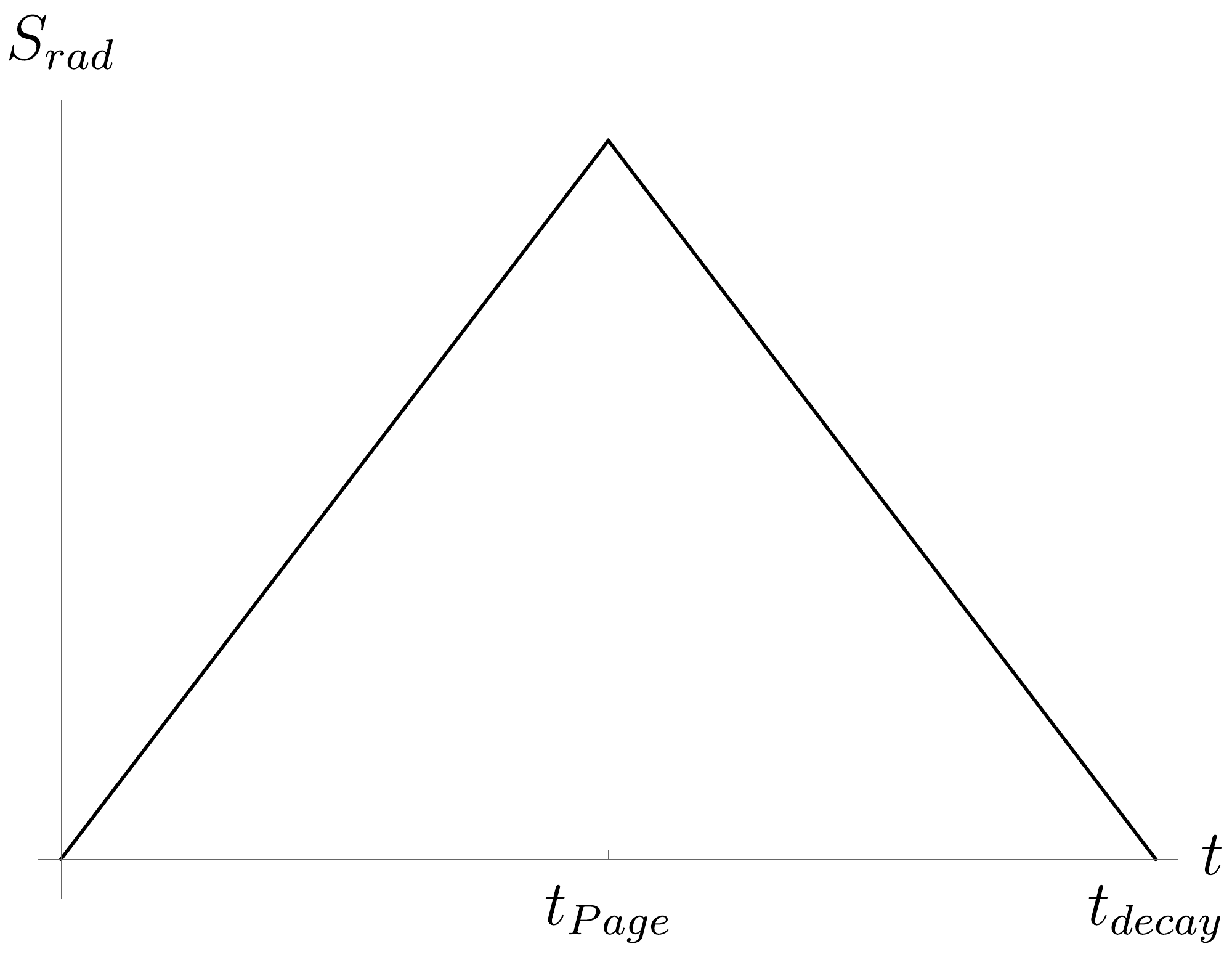}
      &
      \includegraphics[width=0.33\textwidth]{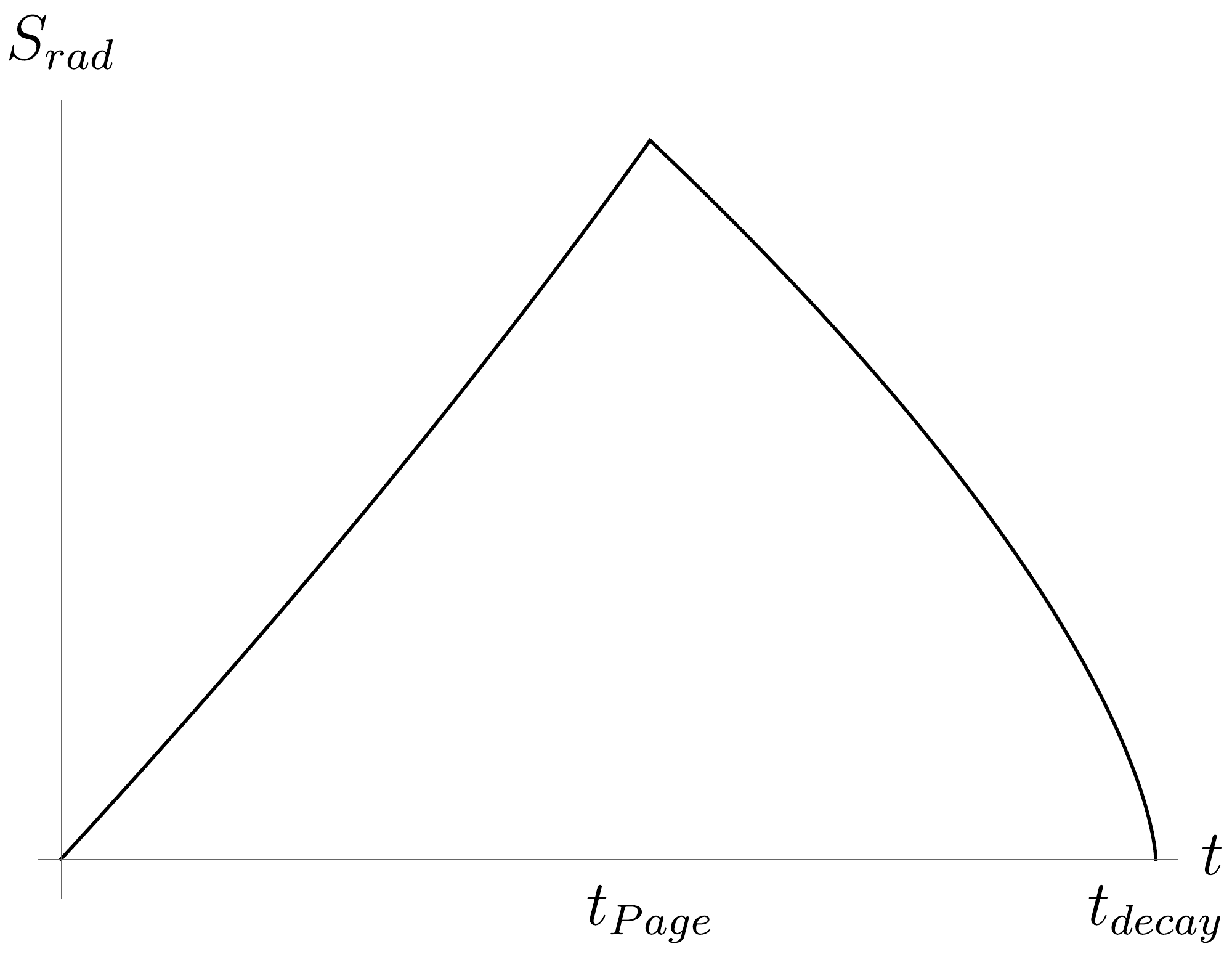}
      &
      \includegraphics[width=0.33\textwidth]{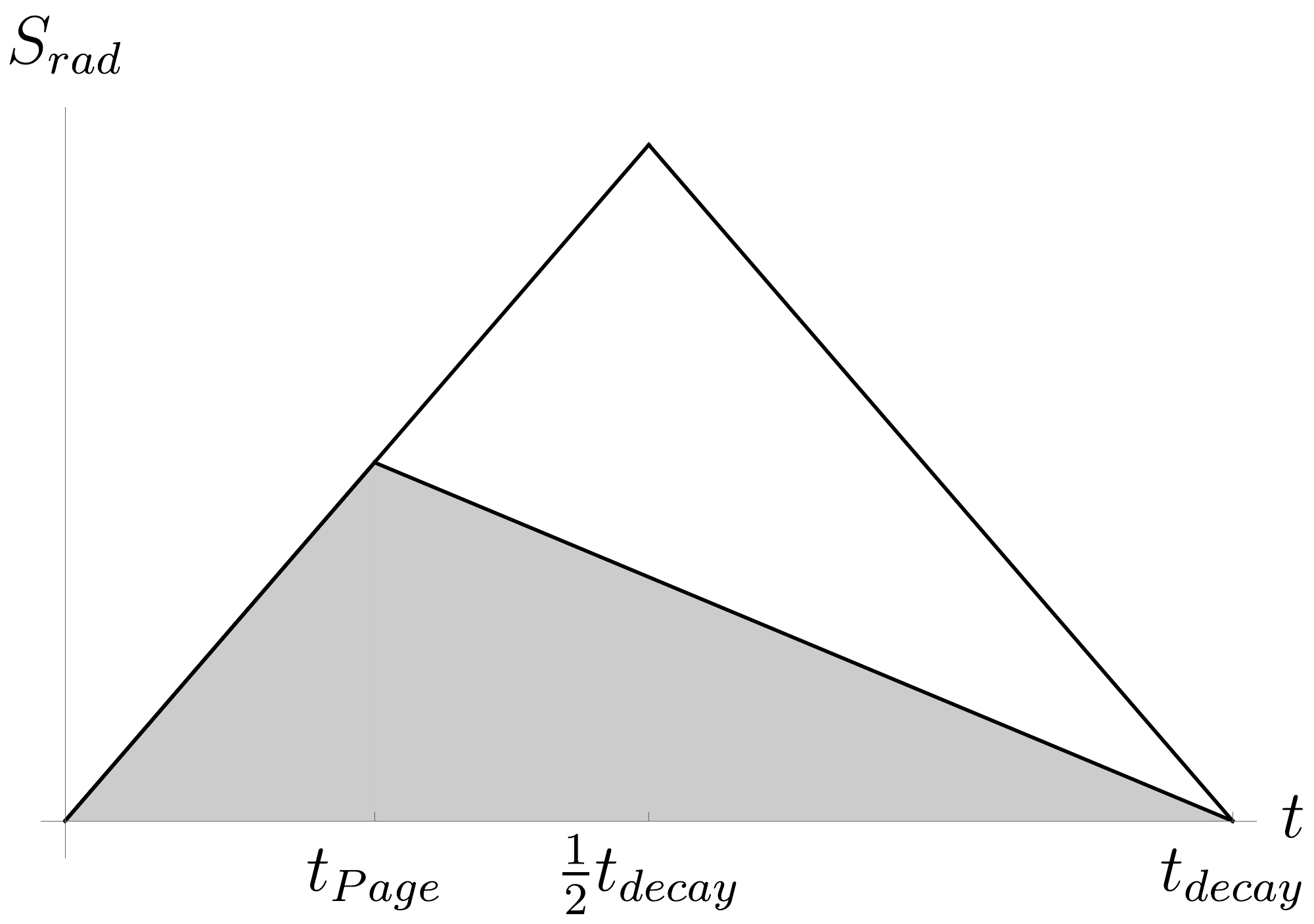}\\    
      \small a) Page Curve &  b) Includes black hole physics & c) Page-like curves\\
\end{tabular}
\caption{}\label{fig1}
\end{figure}

However, in general the Page Curve is not symmetric because the entropy increase of the outgoing radiation, $dS_b/dt$ is not equal to the energy decrease of the black hole, $-dS_{BH}/dt$.  This is due to the different greybody factors, helicities and particle numbers available for massless emission that affect the black hole entropy loss and Hawking radiation entropy differently.  The corresponding Page Curve is shown in Figure \ref{fig1}b. In \cite{Page-3}, using previous scattering calculations, Page showed that 
\[ \beta \equiv \frac{ \; \ dS_b/dt}{-dS_{BH}/dt} \approx 1.484 \ 72 \] 
He also showed that the Page Time, $t_{Page}$ is related to the black hole lifetime $t_{decay}$, as 
\[ t_{Page} = \left[1- \left(\frac{\beta}{\beta +1} \right )^{3/2} \right ]   t_{decay} \approx 0.53 \ t_{decay}  \quad \textrm{for} \quad \beta = 1.484 \ 72.
\]
Thus, the Page Time is heavily dependent on $\beta$.  

    Our result in (\ref{label10}) generates the envelope of any Page-like curve for the evaporation via Bell pair emission model that we have been considering. The monotonically increasing part is bounded by a straight line with a slope equal to the maximum of the upper bound in (\ref{label10}).  This is equal to $\log 2$. Similarly, the decreasing part is also bounded by a straight line and the slope of this line is obtained from the minimum of the lower bound in (\ref{label10}). 
\begin{eqnarray*}
         \text{The lower bound} &=& 1 - 4\epsilon_2^2 - \sqrt{1-\gamma^2} \\
          &=& 1 - 4\epsilon_2^2 - \sqrt{4\epsilon^2 (1 - \epsilon^2) - 4\epsilon_2^2}.
\end{eqnarray*}
    For any value of $\epsilon_2$, this quantity is minimized when $\epsilon^2(1-\epsilon^2)$ is maximized. This occurs when $\epsilon^2(1-\epsilon^2) = \frac{1}{4}$. Therefore, we need to consider the minimum value of the quantity $1 - 4\epsilon_2^2 - \sqrt{1 - 4\epsilon_2^2}$, which is $-\frac{1}{4}$. The Page-like curve that our evaporation model generates is bounded by the grey region in Figure \ref{fig1}c.  Note, in drawing the graph in Figure \ref{fig1}c, we first fixed the evaporation time $t_{decay}$. Then we imposed the two bounds. 
    

We will take the intersection of both bounds in  Figure \ref{fig1}c to be the Page Time in this model.  This occurs at $t_{Page} = 0.2 \ t_{decay}$ and corresponds to $\beta = 6.24$.  This corresponds to $\beta = 6.24$ Such an early Page Time would imply that information quickly comes out of the black hole because the emitted photons/gravitons in the outgoing radiation possess much more entropy -- by a factor of $6.24$ --  than the entropy decrease of the black hole. This seems farfetched and would seem to indicate that evaporation via only Bell pair states doesn't correspond to the physical picture of evaporation in Page's calculations.  It also implies that maximal entanglement between the outgoing radiation and the hole doesn't last very long.  

This picture is at odds with Page's description of unitary black hole evaporation.  In Page's description the Page Time is near the half-way point in the evaporation process. General arguments regarding subsystem entropy transfer support the belief that $t_{Page} \approx \frac{1}{2} t_{decay}$ and reasonable estimates of $\beta$ also support this. Adami and Bradler also found modifications to the Page Curve in their study of black hole dynamics using a trilinear Hamiltonian \cite{adami}. In retrospect, the early turnover of the envelope of our Page Curve may not be surprising because a non-negligible $\epsilon$ and $\epsilon_2$ will presumably create stronger interactions and cause the Page Curve's envelope to turn over more quickly.

Giddings and Shi \cite{Giddings:2012dh} generalized Mathur's model and showed that corrections via Bell pair states only, would not restore unitarity.  Our Page Curve puts contraints on Bell pair emission that are unlikely to be met and thus reaffirms Giddings and Shi's, and Mathur's earlier claim that Bell pairs don't restore unitarity.   

\section{Conclusion} \label{sec-9}

Mathur's analysis of `small' corrections to the leading order Hawking analysis showed that Bell Pair states do not unitarize the black hole evaporation process. In this paper, we reaffirm this claim by analyzing a toy qubit model that incorporates a generalized quantum correlation between successive pairs. We then generalize Mathur's work by relaxing the `smallness' condition. We establish a rigorous nontrivial upper and lower bound for the change in entanglement entropy. This enables us to parameterize the required correction to the Hawking state.  Our results show that the correction required to restore unitarity is `not-so-small' or even `large'. However, even if we allow such an evaporation law, we find that it is at odds with the expected Page Curve. If a black hole's time evolution is to be unitary, the entropy of its outgoing radiation should approximate the expected Page curve. Thus, this leads us to reaffirm the belief that the information paradox cannot be resolved by adhering to the picture of an evaporation process via only Bell Pair states.

\section*{Acknowledgements}
We would like to thank Samir Mathur for helpful suggestions and comments. M.H.R. acknowledges partial support from the U.S. Department of Energy under Grant No. DE-SC0010296. A.R. acknowledges partial support from the U.S. Department of Energy, Office of High Energy Physics under Grant No. DE‐SC0007890

\bibliography{correlation}
\bibliographystyle{apsrev4-1}

\end{document}